\documentclass[conference]{IEEEtran}
\IEEEoverridecommandlockouts

\usepackage{tikz}
\usetikzlibrary{positioning,arrows.meta}

\usepackage[utf8]{inputenc}
\usepackage{mathtools}
\usepackage{blkarray}
\usepackage{cuted}

\usepackage{cite}
\usepackage{amsmath,amssymb,amsfonts}
\allowdisplaybreaks
\usepackage{algorithmic}
\usepackage{graphicx}
\usepackage{textcomp}
\usepackage{xcolor}
\usepackage{multirow}
\usepackage{caption}
\usepackage{subcaption}
\usepackage{comment}
\usepackage{url}
\usepackage{float}
\usepackage[T1]{fontenc}

\usepackage{bm}
\usepackage{bbm}

\usepackage{hyperref}

\usepackage{amsthm}
\newtheorem{theorem}{Theorem}[section]
\newtheorem{corollary}{Corollary}[theorem]
\newtheorem{lemma}[theorem]{Lemma}

\newtheorem{definition}{Definition}
\newtheorem{proposition}[theorem]{Proposition}

\makeatletter
\newsavebox{\@brx}
\newcommand{\llangle}[1][]{\savebox{\@brx}{\(\m@th{#1\langle}\)}%
  \mathopen{\copy\@brx\kern-0.5\wd\@brx\usebox{\@brx}}}
\newcommand{\rrangle}[1][]{\savebox{\@brx}{\(\m@th{#1\rangle}\)}%
  \mathclose{\copy\@brx\kern-0.5\wd\@brx\usebox{\@brx}}}
\makeatother

\title{Analytical Performance Estimations for Quantum Repeater Network Scenarios
\thanks{Authors' email addresses: yzang@uchicago.edu, chungmiranda@anl.gov, kettimut@anl.gov, msuchara@microsoft.com, tzh@uchicago.edu.}}

\author{
\IEEEauthorblockN{Allen Zang\IEEEauthorrefmark{1},
Joaquin Chung\IEEEauthorrefmark{2},
Rajkumar Kettimuthu\IEEEauthorrefmark{2},
Martin Suchara\IEEEauthorrefmark{3},
Tian Zhong\IEEEauthorrefmark{1}
}

\IEEEauthorblockA{\IEEEauthorrefmark{1}University of Chicago, Chicago, IL, USA,
\IEEEauthorrefmark{2}Argonne National Laboratory, Lemont, IL, USA,\\
\IEEEauthorrefmark{3}Microsoft Corporation, Redmond, WA, USA
}}

\def\BibTeX{{\rm B\kern-.05em{\sc i\kern-.025em b}\kern-.08em
    T\kern-.1667em\lower.7ex\hbox{E}\kern-.125emX}}
\begin{document}

\maketitle

\begin{abstract}
    Quantum repeater chains will form the backbone of future quantum networks that distribute entanglement between network nodes. Therefore, it is important to understand the entanglement distribution performance of quantum repeater chains, especially their throughput and latency. By using Markov chains to model the stochastic dynamics in quantum repeater chains, we offer \textit{analytical estimations} for long-run throughput and on-demand latency of continuous entanglement distribution. We first study single-link entanglement generation using general multiheralded protocols. We then model entanglement distribution with entanglement swapping over two links, using either a single- or a double-heralded entanglement generation protocol. We also demonstrate how the two-link results offer insights into the performance of general $2^k$-link nested repeater chains. Our results enrich the quantitative understanding of quantum repeater network performance, especially the dependence on system parameters. The analytical formulae themselves are valuable reference resources for the quantum networking community. They can serve as benchmarks for quantum network simulation validation or as examples of quantum network dynamics modeling using the Markov chain formalism.
\end{abstract}

\section{Introduction}
Quantum networks~\cite{kimble2008quantum,wehner2018quantum} will enable various distributed quantum information processing applications, such as distributed quantum computing~\cite{cacciapuoti2019quantum,cuomo2020towards}, sensing~\cite{zhang2021distributed}, and secure communication~\cite{mehic2020quantum}. This vision of quantum networks has motivated significant experimental efforts~\cite{moehring2007entanglement,ritter2012elementary,hofmann2012heralded,bernien2013heralded,kalb2017entanglement,humphreys2018deterministic,pompili2021realization,hermans2022qubit,pompili2022experimental,knaut2024entanglement,liu2024creation}. 

Long-distance entanglement distribution service is impossible without quantum repeater chains, and it is thus necessary to examine quantum repeater chains' entanglement distribution throughput and latency, because they are two of the most important metrics for practical applications. Although quantum network simulators~\cite{coopmans2021netsquid,wu2021sequence,satoh2022quisp} could in principle evaluate complicated scenarios that are not analytically tractable, analytical formulae are more desirable for simple scenarios: They offer clear, quantitative intuition and insight into the explicit dependence of performance on system parameters, and the formulae can further serve as ground truths for the validation of quantum network simulators, giving us more confidence in their accuracy. In this work we use Markov chain~\cite{mc-razavi2009physical,mc-basu2012percolation,mc-jones2013high,mc-vardoyan2019stochastic,vinay2019statistical,shchukin2019waiting,mc-vardoyan2020on,mc-nain2020analysis,mc-dai2021entanglement,mc-vasantam2021stability,mc-vasantam2022throughput,mc-chandra2022scheduling,mc-nain2022analysis,mc-zubeldia2022matching,mc-coopmans2022improved,mc-shchukin2022optimal,mc-vardoyan2022quantum,mc-vardoyan2023capacity,mc-illiano2023design,inesta2023performance,mc-davies2023entanglement,mc-goodenough2024noise} to easily obtain new instances of exact, analytical formulae for estimating throughput and latency in quantum repeater network scenarios.
We note that some closely related works, for example,~\cite{vinay2019statistical,shchukin2019waiting,mc-coopmans2022improved,inesta2023performance}, focus on detailed evaluations of more complicated scenarios. Our objective is to provide simple, explicit results that offer intuition on the most important and straightforward figures of merit for specific network configurations and protocols. We also note some fundamental limits of quantum communication~\cite{pirandola2017fundamental,pirandola2019end}.
We are also the first to generalize the double-heralded Barrett--Kok protocol (BKP)~\cite{barrett2005efficient} and examine from a networking perspective \textit{multiheralded} entanglement generation protocols.

We assume that the distributed entangled pair will be immediately used or buffered to free up quantum memories, so that entanglement generation and distribution are performed continuously~\cite{chakraborty2019distributed,kolar2022adaptive,inesta2023performance}. Specifically, we estimate in one- and two-link quantum repeater chains (i) the long-run throughput, where the continuous protocol is run for a sufficiently long time, and (ii) the on-demand latency, where the continuous protocol is initiated when a request is received and is terminated when one EPR pair is distributed. Additionally, we show that the two-link results allow us to estimate the mean throughput of $2^k$-link nested repeater chains.

This paper is organized as follows. We review the basics of Markov chains in Sec.~\ref{sec:background} and present the methodology in Sec.~\ref{sec:methodology}. We study the single-link scenario in Sec.~\ref{sec:1link} and the three-node chain in Sec.~\ref{sec:2link}, while we demonstrate the generalization of the two-link results in Sec.~\ref{sec:extend}. In Sec.~\ref{sec:conclusion} we summarize our work. 
\section{Markov chain preliminaries} \label{sec:background}
Here we present key definitions and facts about Markov chains. Proofs are omitted.
\begin{definition}[Markov chain]
    A discrete-time Markov chain is a sequence of random variables $X_0, X_1, X_2, \dots$ that satisfies Markovianity, that is, $P(X_{n+1}=x\vert X_0=x_0,X_1=x_1,\dots,X_n=x_n)=P(X_{n+1}=x\vert X_n=x_n)$, $\forall n\geq 0$, $\forall x\in\mathcal{S}$, where $\mathcal{S}$ is the state space that is a countable set of all possible values of $X_i,\forall i\in\mathbb{Z}^*$.
\end{definition}

A Markov chain can be represented by the \textit{transition matrix}, whose matrix element $P_{ij}$ represents the  probability of transitioning from state $i$ to $j$. According to Markovianity, given the probability distribution of system states $\pi_n$ at step $n$, the distribution at step $n+1$ is $\pi_{n+1}(j) = \sum_i\pi_n(i)P_{ij} \Leftrightarrow \pi_{n+1} = \pi_nP$, with the distributions being row vectors. In practice, the \textit{equilibrium distribution} is of particular interest.
\begin{definition}[Equilibrium]\label{def:equilibrium}
    The equilibrium distribution $\Tilde{\pi}$ of a Markov chain with transition matrix $P$ satisfies $\Tilde{\pi} = \Tilde{\pi}P$.
\end{definition}
We want to know whether a Markov chain has an equilibrium, whether it is unique, and whether it can be achieved. To answer these questions, we need additional formalism.
\begin{definition}[Accessibility]
    A state $j\in\mathcal{S}$ is accessible from state $i\in\mathcal{S}$ (denoted as $i\rightarrow j$) if $\exists t>0$ s.t. $(P^t)_{ij}>0$.
\end{definition}
\begin{definition}[Communication]
    States $i$ and $j$ communicate if $i\rightarrow j$ and $j\rightarrow i$ (denoted as $i\leftrightarrow j$). A communicating class $\mathcal{C}\subset\mathcal{S}$ is a set of states whose members communicate.
\end{definition}
\begin{definition}[Irreducibility]\label{def:ird}
    A Markov chain is irreducible if $i\leftrightarrow j,\ \forall i,j\in\mathcal{S}$.
\end{definition}
\begin{definition}[Periodicity]\label{def:period}
    A state $i$ is periodic if $d(i) = \gcd\{t:(P^t)_{ii}>0\}>1$, and aperiodic otherwise, where $\gcd$ denotes the greatest common divisor.
\end{definition}
Irreducibility and aperiodicity are important for the convergence to equilibrium. The following lemma helps us determine aperiodicity for irreducible chains.
\begin{lemma}
    If $i\leftrightarrow j$, then $d(i)=d(j)$.
\end{lemma}
\begin{corollary}\label{thm:ap_ird}
    An irreducible chain with a state whose one-step transition probability to itself is nonzero is aperiodic.
\end{corollary}
We now present the convergence theorem~\ref{thm:converge} that will be important in studying long-run entanglement distribution.
\begin{theorem}[Convergence theorem, informal]\label{thm:converge}
    An irreducible and aperiodic Markov chain on a finite state space will converge to its unique equilibrium.
\end{theorem}

We go beyond equilibrium and consider the hitting time.
\begin{definition}[Hitting time]
    For a Markov chain $\{X_{i\geq 0}\}$, $t_x=\min\{t\geq 1: X_t=x\}$ is called the hitting time for state $x$.
\end{definition}
For an irreducible and aperiodic Markov chain, we have powerful tools to obtain the mean and the variance of time to hit state $j$ when the chain starts from state $i$.
\begin{definition}[Fundamental matrix]
    For an irreducible and aperiodic Markov chain with $|\mathcal{S}|$ states, the fundamental matrix $N_i$ with respect to  state $i$ is $(I_{|\mathcal{S}|-1}-Q)^{-1}$, where $I_{|\mathcal{S}|-1}$ is an $(|\mathcal{S}|-1)$-dimensional identity matrix and $Q$ is the matrix obtained from removing the column and the row corresponding to state $i$ from the transition matrix.
\end{definition}
\begin{theorem}[Mean hitting time]\label{thm:mean_ht}
    For an irreducible and aperiodic Markov chain, the expected hitting time $\mathbb{E}[t_j\vert X_0=i]=\bm{t}_i$, where $\bm{t}=N_j\cdot\bm{1}$ and $\bm{1}$ is a column vector of length $(|\mathcal{S}|-1)$ with all entries being 1.
\end{theorem}
\begin{theorem}\label{thm:mean_ht_2}
    For an irreducible and aperiodic Markov chain, the expected hitting time for state $i$ when starting from state $i$ is $\mathbb{E}[t_i\vert X_0=i]=1/\Tilde{\pi}_i$, where $\Tilde{\pi}_i$ is the probability of state $i$ in the unique equilibrium distribution.
\end{theorem}
\begin{theorem}[Variance of hitting time]\label{thm:var_ht}
    The variance of hitting time for state $j$ starting from state $i$ in an irreducible and aperiodic Markov chain is the $i$th entry of $(2N_j-I_{|\mathcal{S}|-1})\bm{t} - \bm{t}\odot\bm{t}$, where $\odot$ denotes the Hadamard product.
\end{theorem}

\section{Analytical estimation methodology}\label{sec:methodology}
First we show that the above convergence and hitting time results are applicable to continuous protocols.
\begin{proposition}\label{thm:ir_ap_continuous_protocol}
    Markov chains of continuous protocols are irreducible and aperiodic.
\end{proposition}
\begin{proof}
    In Markov chains for continuous protocols, every state is achievable from any other state. In practice,  there is always a nonzero probability to transition from a state representing failure in one step to itself. 
\end{proof}
\begin{proposition}\label{thm:fail_succ_equal}
    For Markov chains of continuous protocols, the row corresponding to the failure state is identical to that of the success state.
\end{proposition}
\begin{proof}
    Both the success state and the failure state are the start of the next entanglement distribution attempt. Therefore, their transition patterns are identical.
\end{proof}

\subsection{Throughput}
\begin{definition}[Throughput]
    The throughput $T$ is the number of successfully distributed entangled pairs per unit time.
\end{definition}
Suppose entanglement distribution is continuously attempted for $N\gg 1$ time steps. The average throughput is $\bar{T}=\mathbb{E}[n_\mathrm{S}]/(N\tau)$, where $n_\mathrm{S}$ is the number of visits to state $\mathrm{S}$ that represents success and $\tau$ is the assumed real time corresponding to the time step of the chain. If the chain with transition matrix $P$ starts at state $i$, the expected number of times of visiting $\mathrm{S}$ is $\mathbb{E}[n_\mathrm{S}\vert X_0=i]=\mathbb{E}[\sum_{k=1}^N\mathbbm{1}_\mathrm{S}^k\vert X_0=i]=\sum_{k=1}^N\left(P^k\right)_{i\mathrm{S}}$, where $\mathbbm{1}_\mathrm{S}^k$ denotes the indicator function at the $k$-th step that takes value 1 if in state $\mathrm{S}$ and 0 otherwise. Since the chain will converge according to Proposition~\ref{thm:ir_ap_continuous_protocol}, we have $\sum_{k=1}^N\left(P^k\right)_{i\mathrm{S}}\approx N\Tilde{\pi}_\mathrm{S}$, that is, $\bar{T}\approx\Tilde{\pi}_\mathrm{S}/\tau$. Moreover, we have the formal expression of throughput variance:
\begin{align}
    N^2\tau^2\mathrm{Var}[T] =& \sum_{k=1}^N\mathrm{Var}\left[\mathbbm{1}_\mathrm{S}^k\right] + 2\sum_{1\leq k<l\leq N}\mathrm{Cov}\left[\mathbbm{1}_\mathrm{S}^k,\mathbbm{1}_\mathrm{S}^l\right] \nonumber\\
    =& 2\sum_{1\leq k<l\leq N}\left[\left(P^{l-k}\right)_{\mathrm{S}\mathrm{S}}-\left(P^l\right)_{i\mathrm{S}}\right]\left(P^k\right)_{i\mathrm{S}} \nonumber\\
    & +\sum_{k=1}^N \left(P^k\right)_{i\mathrm{S}}\left[1-\left(P^k\right)_{i\mathrm{S}}\right].
\end{align}
Simply treating each time step as independent and identically distributed according to the equilibrium distribution gives the estimation $\mathrm{Var}[T]\approx\Tilde{\pi}_\mathrm{S}(1-\Tilde{\pi}_\mathrm{S})/(N\tau^2)$, and this could potentially result in significant relative estimation error due to covariance terms. Because of the $1/N$ scaling, however, the absolute error will still be small. In summary, the estimation of throughput is generally reduced to the calculation of $\Tilde{\pi}_\mathrm{S}$.

\subsection{Latency}
\begin{definition}[Latency]
    The latency $L$ is the time needed for an entangled pair to be distributed after initialization.
\end{definition}
Suppose entanglement distribution is initiated at the moment when a request for one EPR pair is received. In order to serve the request, entanglement distribution will be continuously attempted until one EPR pair is successfully distributed. The latency is then $L=t_\mathrm{S}\tau$, where $t_\mathrm{S}$ is the time for the chain to hit the success state $\mathrm{S}$ starting from the initial/failure state. Then, the calculation of mean latency and its variance is reduced to the derivation of the fundamental matrix of the Markov chain w.r.t. state $\mathrm{S}$, which is needed for the application of Theorem~\ref{thm:mean_ht} and Theorem~\ref{thm:var_ht}. If we  care only about mean latency, we can use Theorem~\ref{thm:mean_ht_2}, and $\Tilde{\pi}_\mathrm{S}$ is sufficient according to Proposition~\ref{thm:fail_succ_equal}.

\subsection{Example: single-heralded EG}
Assuming each EG attempt has a success probability $p$ and takes time $\tau$ and the entire continuous EG takes $N\tau$, the average throughput is just $\bar{T} = p/\tau$, which is exact $\forall N$ because the transition matrix is idempotent, that is, $P^k=P,\forall k>0$. This property also helps us  obtain the exact throughput variance $\mathrm{Var}[T]=p(1-p)/(N\tau^2)$. The mean latency can be obtained immediately from the equilibrium distribution as $\bar{L}=\tau/p$. A straightforward calculation gives $\mathrm{Var}[L]=(1-p)\tau^2/p^2$.
\section{One-link entanglement generation} \label{sec:1link}
We first analyze double-heralded protocols such as the BKP. Then we study general multiheralded EG that requires $n$ rounds of successful heralding.

\subsection{Double-heralded EG}\label{sec:tp_BKP}
Suppose the success probability of the first round is $p_1$ and for the second round is $p_2$, and both rounds take equal time $\tau$, which will be the time unit. The process can be modeled as a 3-state Markov chain where state 0 represents failure, state 1 represents success in the first round, and state 2 represents success in the second round, that is, the final success of EG:
\begin{enumerate}
    \item If the current state is 0, the next step is the first EG round. The probability to transition to state 1 is $p_1$, and the probability to stay in state 0 is $(1-p_1)$.
    \item If the current state is 1, the next step is the second EG round. The probability to transition to state 2 is $p_2$, and the probability to stay in state 0 is $(1-p_2)$.
    \item If the current state is 2, the next step is the first EG round. The probability to transition to state 1 is $p_1$, and the probability to stay in state 0 is $(1-p_1)$.
\end{enumerate}
The corresponding transition matrix is
\begin{equation}
    P_\mathrm{BKP} = 
    \begin{pmatrix}
        1-p_1 & p_1 & 0 \\
        1-p_2 & 0 & p_2 \\
        1-p_1 & p_1 & 0
    \end{pmatrix},
\end{equation}
from which we obtain the equilibrium probability of state 2, $\Tilde{\pi}_2 = p_1p_2/(1+p_1)$, as is visualized in Fig.~\ref{fig:2h-gen-throughput}. Besides the obvious feature that higher success probabilities lead to higher throughput, the asymmetry demonstrates different contributions from the first and the second heralding rounds. Specifically, the second round's success probability has larger impact, which can be seen from the fact that $p_1p_2/(1+p_1)-p_1p_2/(1+p_2) = p_1p_2(p_2-p_1)/[(1+p_1)(1+p_2)]$ is negative when $p_1>p_2$; that is, $\forall p_1>p_2$,  when the success probability of the first round is $p_1$ and that of the second round is $p_2$, the throughput is lower than if the success probabilities are interchanged. Intuitively, the reason  is that  failure in the second round results in the ``waste'' of two steps, whereas failure in the first round  ``wastes'' only one. 

\begin{figure}[t]
    \centering
    \begin{subfigure}[b]{0.49\columnwidth}
        \centering
        \includegraphics[width=\textwidth]{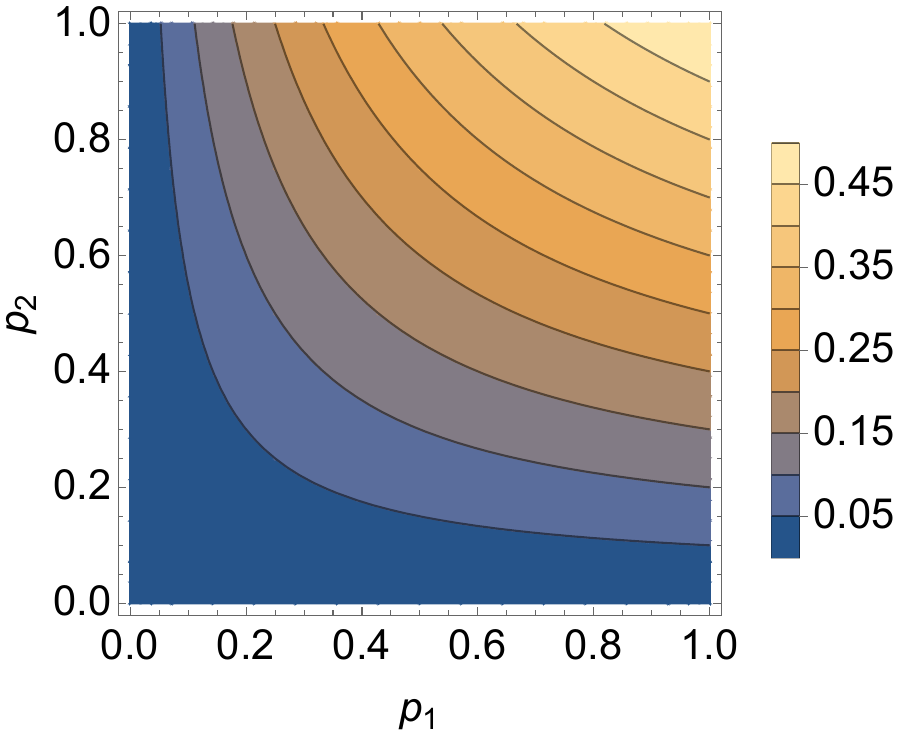}
        \caption{$\Tilde{\pi}_2$ as a function of $p_1$ and $p_2$.}
        \label{fig:2h-gen-throughput}
    \end{subfigure}
    \hfill
    \begin{subfigure}[b]{0.49\columnwidth}
        \centering
        \includegraphics[width=\textwidth]{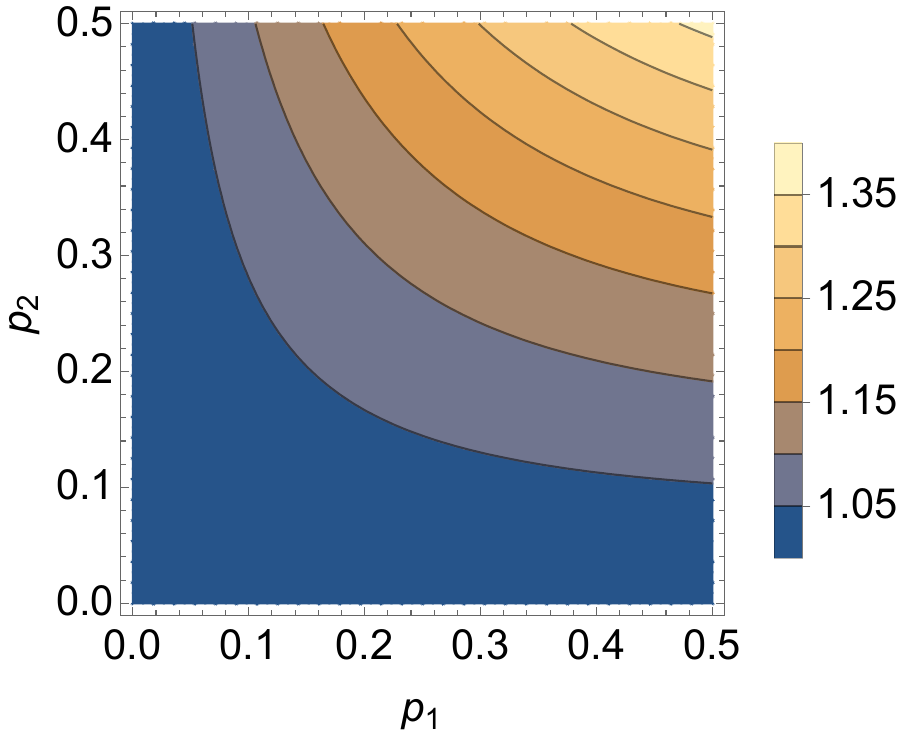}
        \caption{na\"ive vs. exact variance ratio.}
        \label{fig:2h-gen-var-comparison}
    \end{subfigure}
    \caption{Visualization of double-heralded EG throughput results.}
\end{figure}

For latency, we calculate the fundamental matrix w.r.t. state 2, from which we obtain the latency variance:
\begin{equation}
    \mathrm{Var}[L_\mathrm{BKP}]/\tau^2 = \left(\frac{1+p_1}{p_1p_2}\right)^2 - \frac{3+p_1}{p_1p_2}.
\end{equation}
We also visualize the mean latency and the ratio between standard deviation of latency and the mean latency in Fig.~\ref{fig:2h-gen-latency} and Fig.~\ref{fig:2h-gen-latency-std}, respectively.

Fortunately, we can analytically derive $P_\mathrm{BKP}^n$, from which we obtain the exact mean throughput:
\begin{equation}
    \tau\Bar{T}_\mathrm{BKP} = \frac{p_1p_2}{1+p_1} - \left[\frac{p_1p_2}{(1+p_1)^2}\frac{1-(-p_1)^N}{N}\right],
\end{equation}
which as expected converges to $\Tilde{\pi}_2/\tau$ in the large $N$ limit. We can also obtain the exact throughput variance $\forall N$, from which we focus on the leading order term:
\begin{equation}
    N\tau^2\mathrm{Var}[T_\mathrm{BKP}] \approx \frac{p_1p_2\left[(1+p_1)^2-p_1p_2(3+p_1)\right]}{(1+p_1)^3}.
\end{equation}
We then visualize the ratio between the na\"ive estimation $\Tilde{\pi}_2(1-\Tilde{\pi}_2)$ and the right-hand side of the above equation in Fig.~\ref{fig:2h-gen-var-comparison}. We observe that the na\"ive estimation offers high relative accuracy when $p_1,p_2$ are small. The reason is that with small success probabilities, most time steps correspond to independent attempts, whereas with high success probabilities, there are more consecutive events that are correlated.

\begin{figure}[t]
    \centering
    \begin{subfigure}[b]{0.49\columnwidth}
        \centering
        \includegraphics[width=\textwidth]{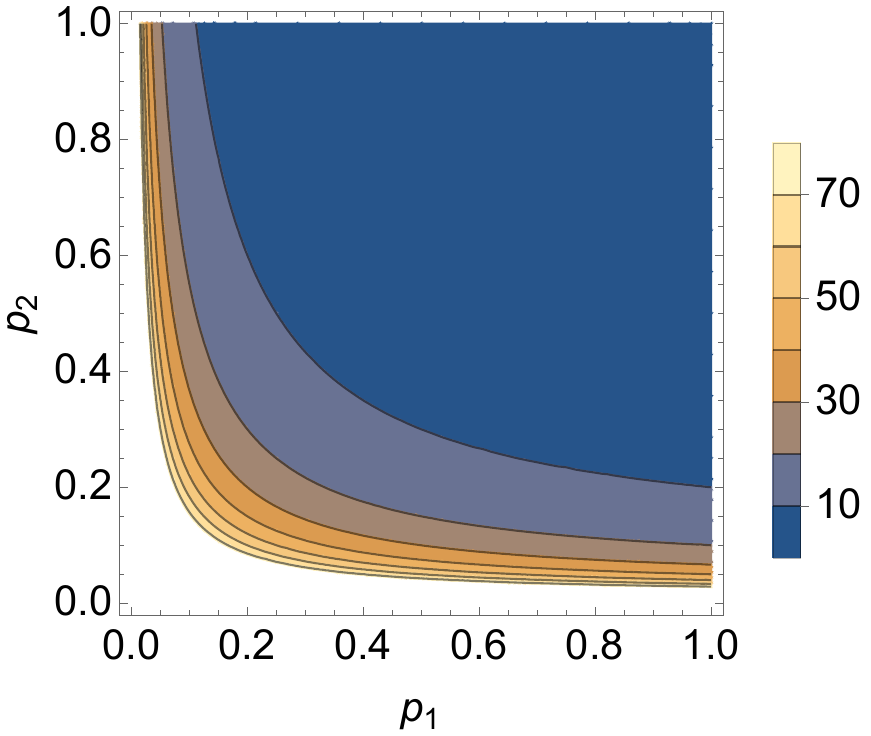}
        \caption{Mean latency (in $\tau$).}
        \label{fig:2h-gen-latency}
    \end{subfigure}
    \hfill
    \begin{subfigure}[b]{0.49\columnwidth}
        \centering
        \includegraphics[width=\textwidth]{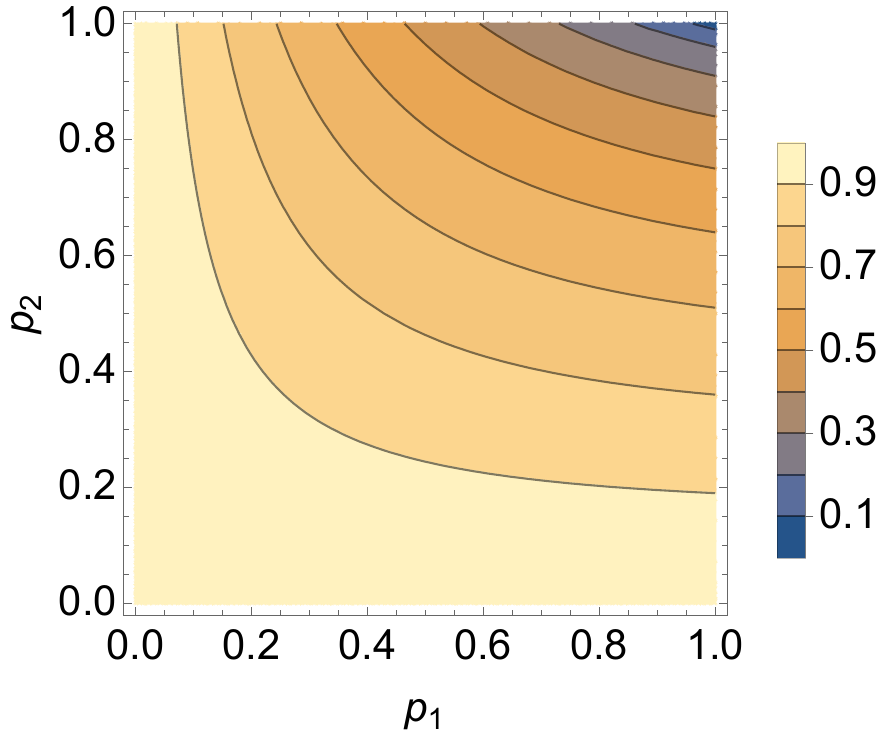}
        \caption{Std. vs. mean latency ratio.}
        \label{fig:2h-gen-latency-std}
    \end{subfigure}
    \caption{Visualization of double-heralded EG latency results.}
\end{figure}

\subsection{Generalization to multiheralded EG}
\label{sec:general}
It is natural to generalize the double-heralded EG to modeling multiheralded EG, whose overall success requires success in all $n$ rounds, with success probabilities of each round being $p_i,\ i=1,2,\dots,n$. Such protocols using multiple heralding rounds to ensure the correct generation of expected entangled state could offer higher entanglement fidelity conditioned on success at the cost of lower success rate. We again assume each round has identical duration $\tau$. Now the Markov chain contains $(n+1)$ states, where state $i$ corresponds to success of the $i$th round with $i=1,2,\dots,n$ and state 0 corresponds to failures. The possible state transitions are as follows:
\begin{enumerate}
    \item If the current state is $i$ ($i<n$), the next step is the $(i+1)$th EG round. The probability to transition to state $(i+1)$ is $p_{i+1}$, and the probability to transition to state 0 is $(1-p_{i+1})$.
    \item If the current state is $n$, the next step is the first EG round. The probability to transition to state 1 is $p_1$, and the probability to transition to state 0 is $(1-p_1)$.
\end{enumerate}
The corresponding transition matrix is
\begin{equation}
    P_n = 
    \begin{pmatrix}
        1-p_1 & p_1 & 0 & \dots & 0 \\
        1-p_2 & 0 & p_2 & \dots & 0 \\
        \vdots & \vdots & \vdots & \ddots & \vdots \\
        1-p_n & 0 & 0 & \dots & p_n \\
        1-p_1 & p_1 & 0 & \dots & 0 \\
    \end{pmatrix}.
\end{equation}
We obtain the equilibrium probability of state $n$, $\Tilde{\pi}_n$, by solving the left eigen equation of $P_n$:
\begin{equation}
    \Tilde{\pi}_n = \frac{\prod_{i=1}^np_i}{1 + \sum_{i=1}^{n-1}\prod_{j=1}^ip_j},
\end{equation}
from which we can estimate the mean and the variance of the throughput and the mean latency. One can easily verify that the results for double-heralded EG protocols are compatible with this general expression. In fact, we can analytically obtain the variance of latency using Theorem~\ref{thm:var_ht}:
\begin{equation}
\begin{aligned}
    \mathrm{Var}[L_n]/\tau^2 =& \left(\frac{1 + \sum_{i=1}^{n-1}\prod_{j=1}^ip_j}{\prod_{i=1}^np_i}\right)^2\\
    &+ \frac{2n-1 + \sum_{i=1}^{n-1}(2i-1)\prod_{j=1}^{n-i}p_j}{\prod_{i=1}^np_i}.
\end{aligned}
\end{equation}
\section{Two-link entanglement distribution} \label{sec:2link}
When two quantum network nodes do not have a direct connection, entanglement swapping is necessary. Here we study entanglement distribution in a 3-node quantum repeater chain, requiring EG on two links and ES on the middle node. We consider both single-heralded and double-heralded EG.

We assume the distance from both end nodes to the middle node is $L$. Thus each EG round has identical duration $\tau\approx L/c$. The time needed for ES is also roughly $\tau$ because of classical communication from the middle node to both end nodes. We assume that both end nodes have one quantum memory and the middle node contains two quantum memories. 

\subsection{Single-heralded EG}
\begin{figure}[t]
    \centering
    \begin{subfigure}[b]{0.49\columnwidth}
        \centering
        \includegraphics[width=\textwidth]{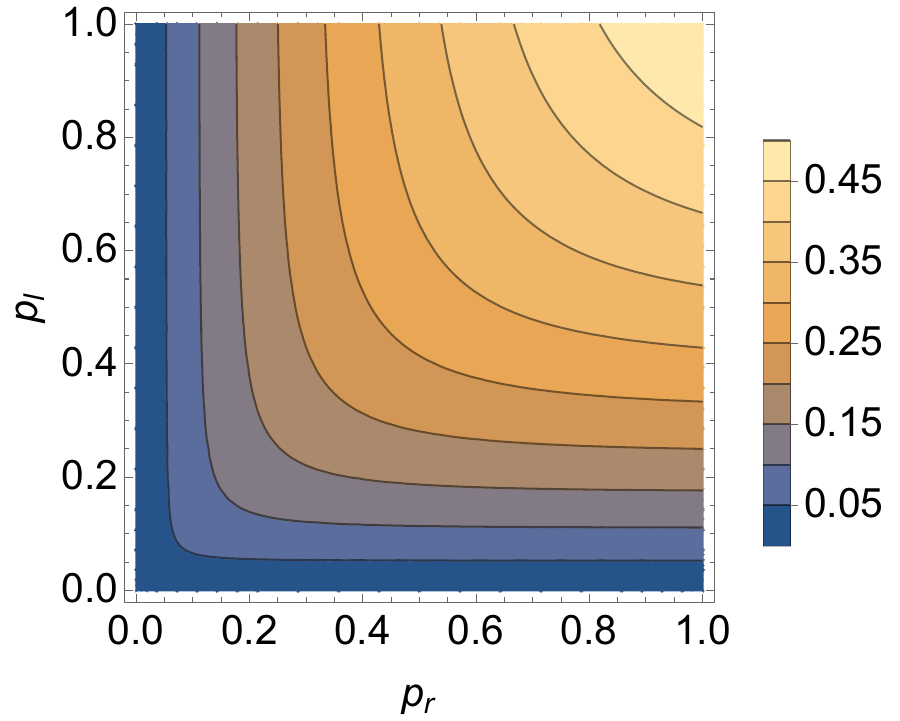}
        \caption{$\Tilde{\pi}_\mathrm{S}$ as function of $p_l$ and $p_r$.}
        \label{fig:1h-swap}
    \end{subfigure}
    \hfill
    \begin{subfigure}[b]{0.49\columnwidth}
        \centering
        \includegraphics[width=\textwidth]{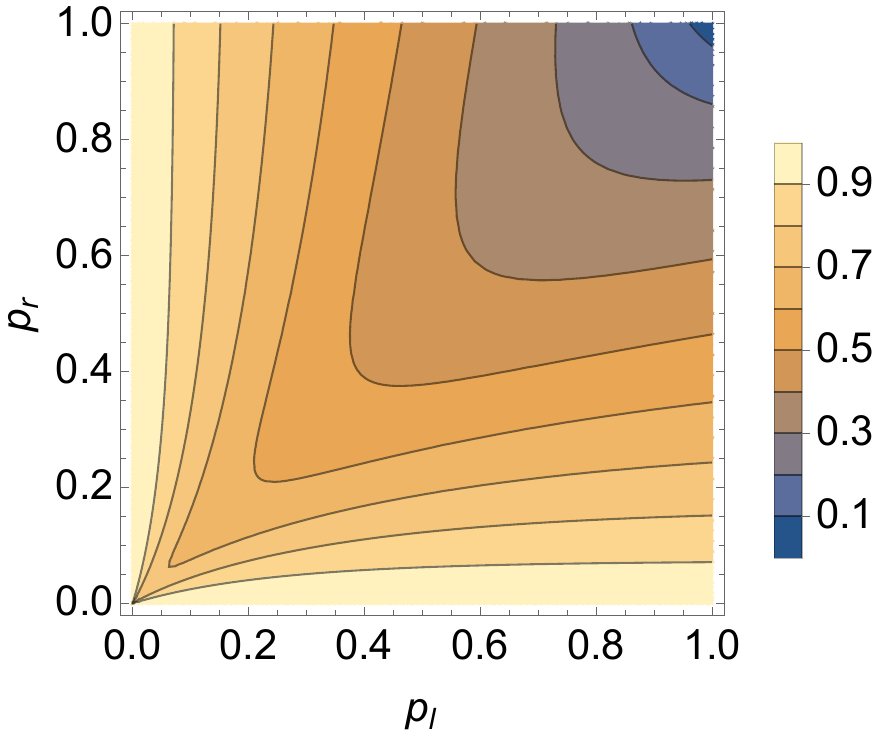}
        \caption{Std. vs. mean latency ratio.}
        \label{fig:1h-swap-latency-std}
    \end{subfigure}
    \caption{Visualization of results for two-link ES with single-heralded EG when $p_s=1$.}
\end{figure}

We assume the success probability for EG on the left link is $p_l$ and on the right link is $p_r$. We further consider the success probability $p_s$ of ES. The possible system states are as follows:
\begin{enumerate}
    \item State 00: both links with no EPR pair, or failure of ES.
    \item State 01: the left link with no EPR pair, the right link with an EPR pair.
    \item State 10: the left link with an EPR pair, the right link with no EPR pair.
    \item State 11: the left link with an EPR pair, the right link with an EPR pair.
    \item State S: ES success.
\end{enumerate}
Then the possible transitions between the states are as follows:
\begin{enumerate}
    \item If the current state is 00, the next step is EG on both links. The probability to stay in state 00 is $(1-p_l)(1-p_r)$, the probability to transition to state 01 is $(1-p_l)p_r$, the probability to transition to state 10 is $p_l(1-p_r)$, and the probability to transition to state 11 is $p_lp_r$.
    \item If the current state is 01, the next step is EG on the left link. The probability to stay in state 01 is $(1-p_l)$, and the probability to transition to state 11 is $p_l$.
    \item If the current state is 10, the next step is EG on the right link. The probability to stay in state 10 is $(1-p_r)$, and the probability to transition to state 11 is $p_r$.
    \item If the current state is 11, the next step is ES. The probability to transition to state 00 is $(1-p_s)$, and the probability to transition to state S is $p_s$.
    \item If the current state is S, the next step is EG on both links. The probability to transition to state 00 is $(1-p_l)(1-p_r)$, the probability to transition to state 01 is $(1-p_l)p_r$, the probability to transition to state 10 is $p_l(1-p_r)$, and the probability to transition to state 11 is $p_lp_r$.
\end{enumerate}
We can write the transition matrix as
\begin{equation}
    P_\mathrm{shs} = 
    \begin{pmatrix}
        \bar{p}_l\bar{p}_r & \bar{p}_lp_r & p_l\bar{p}_r & p_lp_r & 0\\
        0 & \bar{p}_l & 0 & p_l & 0\\
        0 & 0 & \bar{p}_r & p_r & 0\\
        \bar{p}_s & 0 & 0 & 0 & p_s\\
        \bar{p}_l\bar{p}_r & \bar{p}_lp_r & p_l\bar{p}_r & p_lp_r & 0\\
    \end{pmatrix},
\end{equation}
where the subscript ``shs'' denotes single-heralded EG and swapping and $\bar{p}_{l(r,s)}=1-p_{l(r,s)}$. The rows and columns of the transition matrix are both in the order of state 00, 01, 10, 11, and S. Then we can obtain the equilibrium probability of state S:
\begin{equation}\label{eqn:throughput_shs}
    \Tilde{\pi}_\mathrm{S,shs} = \frac{p_lp_rp_s[1-\bar{p}_l\bar{p}_r]}{2p_lp_r + \bar{p}_lp_r^2 + \bar{p}_rp_l^2 - p_lp_r\bar{p}_l\bar{p}_r},
\end{equation}
which is visualized as a function of $p_l$ and $p_r$ when $p_s=1$ in Fig.~\ref{fig:1h-swap}. The throughput is symmetric for $p_1$ and $p_2$ because in this scenario EG on the two links is independent, which is different from the two rounds of double-heralded EG. We also observe that, similar to Fig.~\ref{fig:2h-gen-throughput}, when one success probability is small, increasing the other will not significantly boost the throughput, which manifests the bottleneck effect.

We can further analytically derive the latency variance:
\begin{equation}
\begin{aligned}
    &\mathrm{Var}[L_\mathrm{shs}] = \left(\frac{p_lp_r + p_l^2 + p_r^2 + p_l^2p_r^2}{p_lp_rp_s(p_l + p_r - p_lp_r)}\right)^2\\
    &-\frac{\left(
    \begin{aligned}
        &p_l^2p_r(1-p_r)(2-p_r) + p_l^3(1-p_r)^2(3+p_r)\\
        &+ 3p_r^3 + p_lp_r(4+2p_r-5p_r^2)
    \end{aligned}
    \right)}{p_lp_rp_s(p_l + p_r - p_lp_r)^2}.
\end{aligned}
\end{equation}
We visualize the ratio between the standard deviation of latency $\sqrt{\mathrm{Var}[L_\mathrm{shs}]}$ and the mean latency $\tau/\Tilde{\pi}_\mathrm{S,shs}$ in Fig.~\ref{fig:1h-swap-latency-std}. We note that except for very high success probabilities on both links,  fixing one success probability and increasing the other could result in higher \textit{relative} fluctuations in latency, while the absolute value of the mean latency still drops.

\subsection{Double-heralded EG}
\begin{figure}[t]
    \centering
    \begin{subfigure}[b]{0.49\columnwidth}
        \centering
        \includegraphics[width=\textwidth]{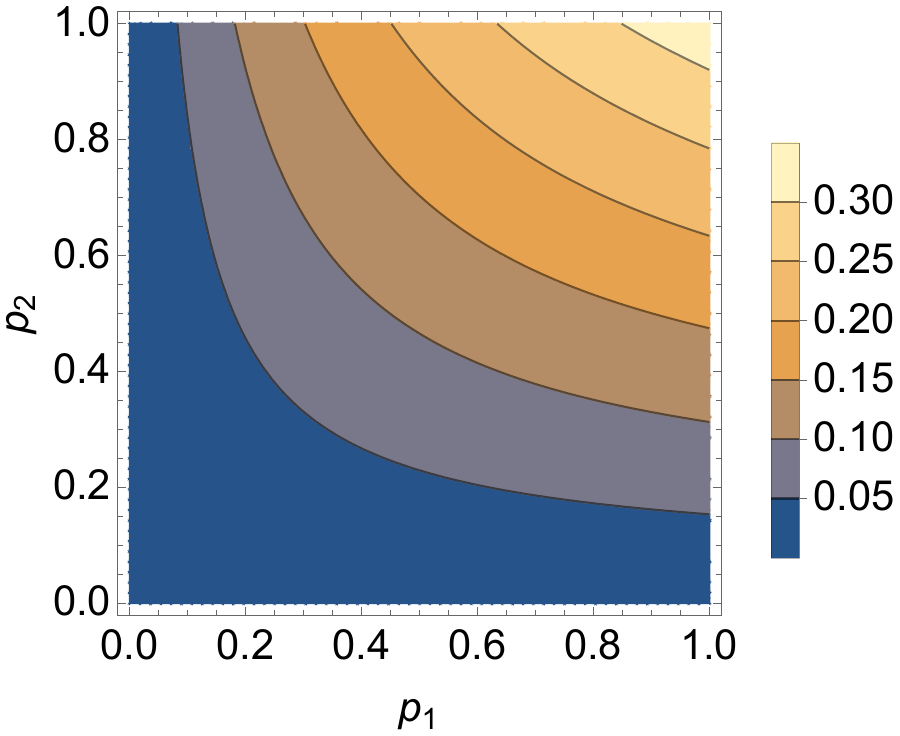}
        \caption{$\Tilde{\pi}_\mathrm{S}$, $p_{r1(2)}=p_{l1(2)}=p_{1(2)}$.}
        \label{fig:2h-swap-12}
    \end{subfigure}
    \hfill
    \begin{subfigure}[b]{0.49\columnwidth}
        \centering
        \includegraphics[width=\textwidth]{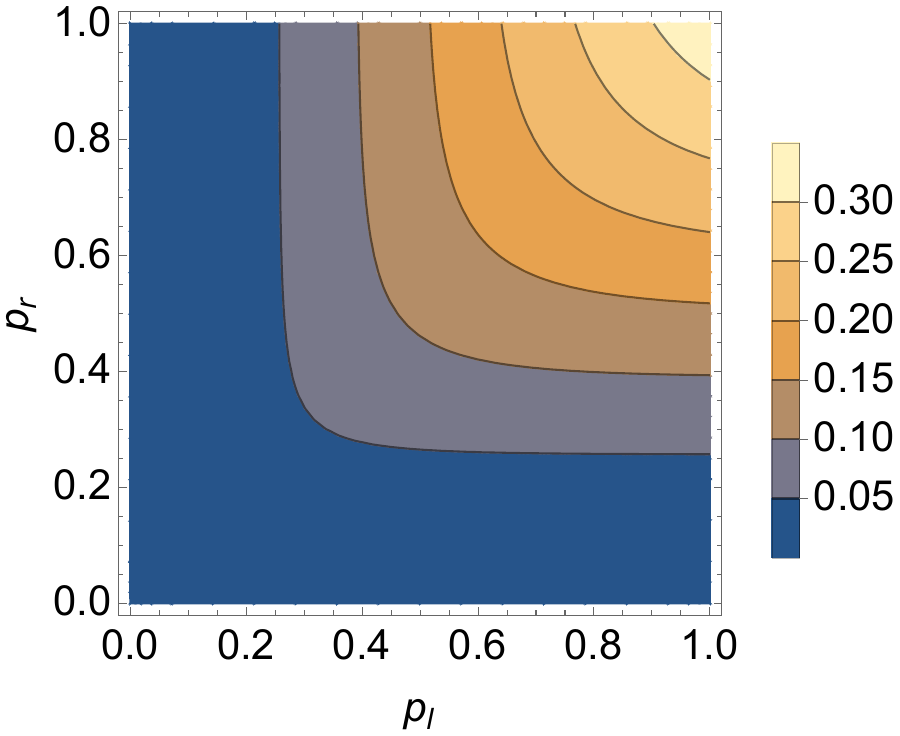}
        \caption{$\Tilde{\pi}_\mathrm{S}$, $p_{r(l)1}=p_{r(l)2}=p_{r(l)}$.}
        \label{fig:2h-swap-lr}
    \end{subfigure}
    \hfill
    \begin{subfigure}[b]{0.49\columnwidth}
        \centering
        \includegraphics[width=\textwidth]{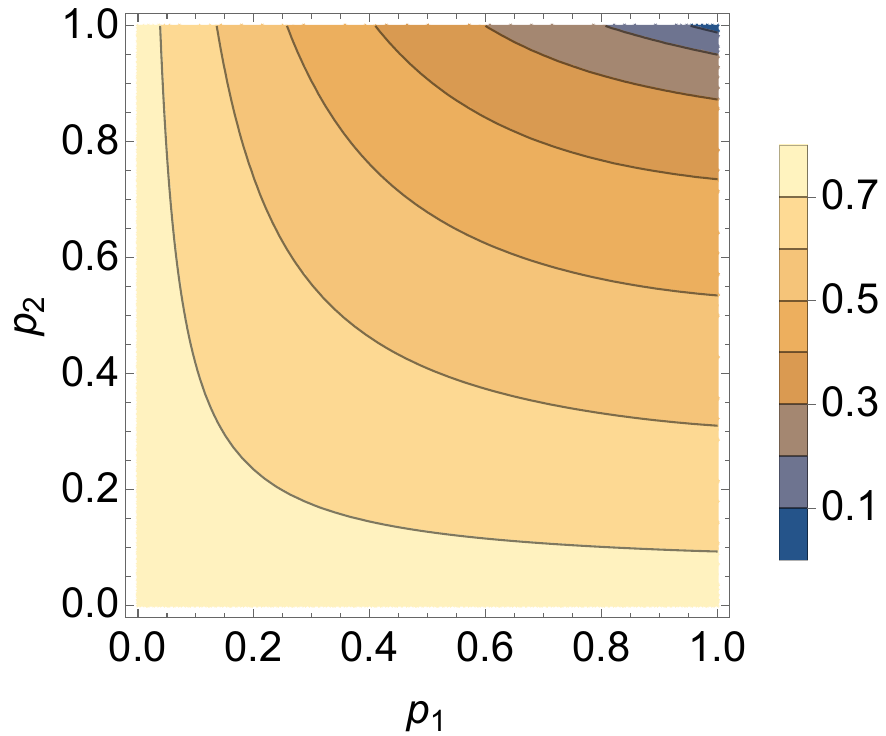}
        \caption{Std. vs. mean latency ratio, $p_{r1(2)}=p_{l1(2)}=p_{1(2)}$.}
        \label{fig:2h-swap-latency-std-12}
    \end{subfigure}
    \hfill
    \begin{subfigure}[b]{0.49\columnwidth}
        \centering
        \includegraphics[width=\textwidth]{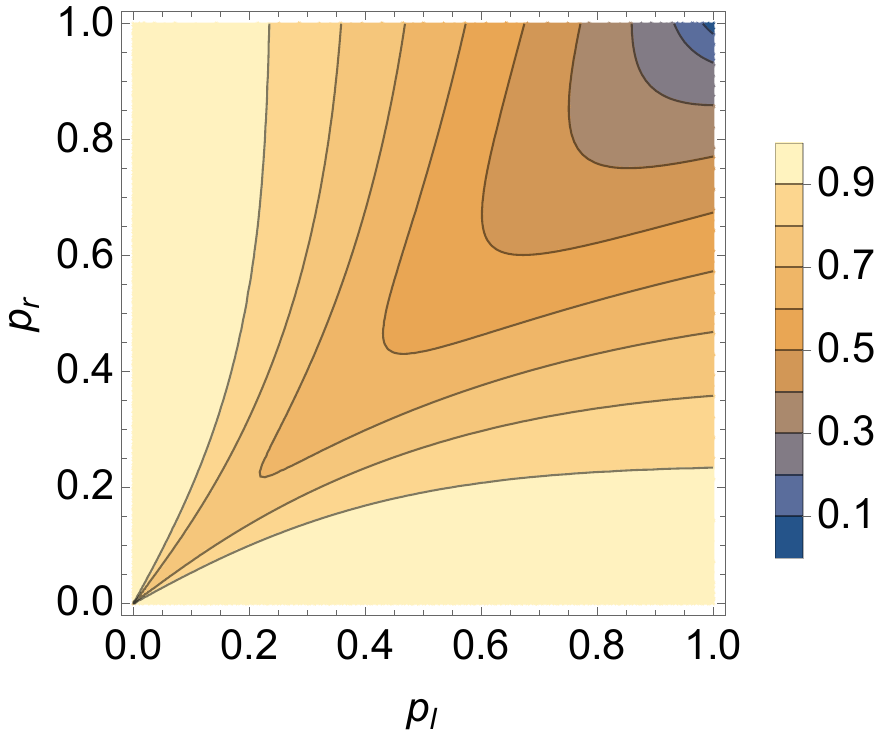}
        \caption{Std. vs. mean latency ratio, $p_{r(l)1}=p_{r(l)2}=p_{r(l)}$.}
        \label{fig:2h-swap-latency-std-lr}
    \end{subfigure}
    \caption{Visualization of results for two-link ES with double-heralded EG when $p_s=1$.}
\end{figure}

\begin{table*}    
\begin{equation}\label{eqn:double-heralded-pmat}
    P_\mathrm{dhs} = 
    \begin{blockarray}{ccccccccccc}
    00 & 01 & 02 & 10 & 11 & 12 & 20 & 21 & 22 & \mathrm{S} \\
    \begin{block}{(cccccccccc)c}
        \bar{p}_{l1}\bar{p}_{r1} & \bar{p}_{l1}p_{r1} & 0 & p_{l1}\bar{p}_{r1} & p_{l1}p_{r1} & 0 & 0 & 0 & 0 & 0 & 00 \\
        \bar{p}_{l1}\bar{p}_{r2} & 0 & \bar{p}_{l1}p_{r2} & p_{l1}\bar{p}_{r2} & 0 & p_{l1}p_{r2} & 0 & 0 & 0 & 0 & 01 \\
        0 & 0 & \bar{p}_{l1} & 0 & 0 & p_{l1} & 0 & 0 & 0 & 0 & 02 \\
        \bar{p}_{l2}\bar{p}_{r1} & \bar{p}_{l2}p_{r1} & 0 & 0 & 0 & 0 & p_{l2}\bar{p}_{r1} & p_{l2}p_{r1} & 0 & 0 & 10 \\
        \bar{p}_{l2}\bar{p}_{r2} & 0 & \bar{p}_{l2}p_{r2} & 0 & 0 & 0 & p_{l2}\bar{p}_{r2} & 0 & p_{l2}p_{r2} & 0 & 11 \\
        0 & 0 & \bar{p}_{l2} & 0 & 0 & 0 & 0 & 0 & p_{l2} & 0 & 12 \\
        0 & 0 & 0 & 0 & 0 & 0 & \bar{p}_{r1} & p_{r1} & 0 & 0 & 20 \\
        0 & 0 & 0 & 0 & 0 & 0 & \bar{p}_{r2} & 0 & p_{r2} & 0 & 21 \\
        1-p_s & 0 & 0 & 0 & 0 & 0 & 0 & 0 & 0 & p_s & 22\\
        \bar{p}_{l1}\bar{p}_{r1} & \bar{p}_{l1}p_{r1} & 0 & p_{l1}\bar{p}_{r1} & p_{l1}p_{r1} & 0 & 0 & 0 & 0 & 0 & \mathrm{S} \\
    \end{block}
    \end{blockarray}
\end{equation}
\end{table*}
We still assume that each EG round and ES all take identical time $\tau$. Now the possible system states are as follows:
\begin{enumerate}
    \item State $ij$, $i,j=0,1,2$: the $i$th EG round on the left succeeds, and the $j$th EG round on the right succeeds.
    \item State S: ES success.
\end{enumerate}
The success probability of the left (right) link's $i$th EG round is $p_{l(r),i}$. The possible transitions between them are as follows:
\begin{enumerate}
    \item If the current state is $ij$, $i,j=0,1$, the next step is the $(i+1)$th EG round on the left link and the $(j+1)$th EG round on the right link. The probability to transition to state 00 is $(1-p_{l,i+1})(1-p_{r,j+1})$, the probability to transition to state $(i+1)0$ is $p_{l,i+1}(1-p_{r,j+1})$, the probability to transition to state $0(j+1)$ is $(1-p_{l,i+1})p_{r,j+1}$, and the probability to transition to state $(i+1)(j+1)$ is $p_{l,i+1}p_{r,j+1}$.
    \item If the current state is $i2$, $i=0,1$, the next step is the $(i+1)$th EG round on the left link. The probability to transition to state 02 is $(1-p_{l,i+1})$, and the probability to transition to state $(i+1)2$ is $p_{l,i+1}$.
    \item If the current state is $2j$, $j=0,1$, the next step is the $(j+1)$th EG round on the right link. The probability to transition to state 20 is $(1-p_{r,j+1})$, and the probability to transition to state $2(j+1)$ is $p_{r,j+1}$.
    \item If the current state is 22, the next step is ES. The probability to transition to state 00 is $(1-p_s)$, and the probability to transition to state S is $p_s$.
    \item If the current state is S, the next step is EG on both links. The probability to transition to state 00 is $(1-p_{l,1})(1-p_{r,1})$, the probability to transition to state 01 is $(1-p_{l,1})p_{r,1}$, the probability to transition to state 10 is $p_{l,1}(1-p_{r,1})$, and the probability to transition to state 11 is $p_{l,1}p_{r,1}$.
\end{enumerate}
The transition matrix is shown in Eqn.~\ref{eqn:double-heralded-pmat}, where the subscript ``dhs'' denotes double-heralded EG and swapping, and $\Bar{p}_{m,i}=1-p_{m,i}$ with $m=l,r$, $i=1,2$. From the transition matrix we can obtain the equilibrium probability of state S, $\Tilde{\pi}_\mathrm{S,dhs}$. While we omit its explicit expression, we visualize it for $p_{r1}=p_{l1}=p_1$ and $p_{r2}=p_{l2}=p_2$ in Fig.~\ref{fig:2h-swap-12} and for $p_{r1}=p_{r2}=p_r$ and $p_{l1}=p_{l2}=p_l$ in Fig.~\ref{fig:2h-swap-lr}, both with $p_s=1$. We see similar bottleneck effects as before, and also the symmetry and asymmetry, which reveal independence between two links, and different contributions from two heralding rounds, respectively. We can also derive the latency variance $\mathrm{Var}[L_\mathrm{dhs}]/\tau^2$ from the fundamental matrix. We visualize the ratio between the latency standard deviation $\sqrt{\mathrm{Var}[L_\mathrm{dhs}]}$ and the mean latency $\tau/\Tilde{\pi} _\mathrm{S,dhs}$, for $p_{r1}=p_{l1}=p_1$ and $p_{r2}=p_{l2}=p_2$ in Fig.~\ref{fig:2h-swap-latency-std-12} and for $p_{r1}=p_{r2}=p_r$ and $p_{l1}=p_{l2}=p_l$ in Fig.~\ref{fig:2h-swap-latency-std-lr}, both with $p_s=1$. One can see that they share similar features with Fig.~\ref{fig:2h-gen-latency-std} and Fig.~\ref{fig:1h-swap-latency-std}, respectively.

\section{Generalization of results}\label{sec:extend}
Generalizations to more complicated scenarios, including multiple quantum memories per node, longer repeater chains, and quantum memory decoherence, must include more possible system states, and transitions between system states will also be more complicated. The number of system parameters will increase as well. Therefore, the analytical intuition of analytical formulae of statistical figures of merit such as throughput and latency will be blurred. Nevertheless, in this section we show how the previous two-link results offer insights into extended scenarios with more links, without specific Markov chain formulation. Specifically, we demonstrate the estimation of long-run throughput in $k$-level ($2^k$-link) repeater chains with nested ES and single-heralded EG~\cite{briegel1998quantum,santra2019quantum,zang2023entanglement}, where all intermediate repeater nodes have two memories and two end nodes have one memory each. For simplicity, we assume that the EG success probability for each elementary link is $p$ and ES deterministically succeeds, that is, $p_s=1$. We choose the elementary link 1-way communication time $\tau$ as the time unit, and the throughput is measured in $1/\tau$. 

Based on the nested structure, we study two alternative types of estimation by recursively using Eqn.~\ref{eqn:throughput_shs}:
\begin{enumerate}
    \item $\bar{T}_k\approx\Tilde{\pi}_\mathrm{S,shs}(2^{k-1}\bar{T}_{k-1},2^{k-1}\bar{T}_{k-1})/2^{k-1},~ k>1$,
    \item $\bar{T}_k\approx\Tilde{\pi}_\mathrm{S,shs}(\bar{T}_{k-1},\bar{T}_{k-1}),~ k>1$.
\end{enumerate}
We take $\Tilde{\pi}_\mathrm{S,shs}$ as a bivariate function by fixing $p_s=1$ in the trivariate Eqn.~\ref{eqn:throughput_shs}, and in both cases we take $T_1=\Tilde{\pi}_\mathrm{S,shs}(p,p)$. The central idea for such estimations is that entanglement distribution in the $(k-1)$-level repeater chain can be effectively viewed as 1-link EG for the highest ($k$th) level ES in the $k$-level repeater chain. The first way takes into account that the time needed for the $k$th level ES is $2^{k-1}$ times the  time unit. Thus when taking the $k$th level ES time as the time unit, the value of throughput for $(k-1)$th level ES needs to be multiplied by $2^{k-1}$, which gives $\Tilde{\pi}_\mathrm{S,shs}(2^{k-1}\bar{T}_{k-1},2^{k-1}\bar{T}_{k-1})$ throughput for the $k$th level. In the end, we need to transform the time unit back to $\tau$, which gives the $2^{k-1}$ denominator. In contrast, the second way of estimation is a simplified version that is dedicated to manifest the abstract recursive structure.

We plot simulation\footnote{https://github.com/allenyxzang/Nested\_Repeater\_Chain} results and analytical estimations of mean throughput for nested repeater chains with $k=2,3,4$ levels in Fig.~\ref{fig:nested_est}, where type (1) estimations are in solid lines and type (2) in dashed lines, while the simulation results are points. Simulation correctness is verified by comparison with analytical results for the two-link scenario. We run the simulation for 100,000 time steps to get one trajectory, and for each pair of $k$ and $p$ we simulate 1,000 trajectories to obtain statistics. The variance is not depicted because it is too small to be shown. We observe that both types of estimations demonstrate qualitative features that are aligned with simulation results, while quantitatively higher repeater chain levels have higher errors, which are due to error accumulation over more rounds of recursion. The na\"ive approach (2) provides surprisingly accurate estimation, as a result of two sources of error offsetting each other: the coarse-grained effective one-link EG overlooks detailed correlations and thus underestimates throughput, while ignoring the ES time overhead overestimates throughput. Our results demonstrate the effectiveness of using prototypical analytical results to obtain quantitative intuition for more complicated scenarios.

\begin{figure}[t]
    \centering
    \includegraphics[width=\linewidth]{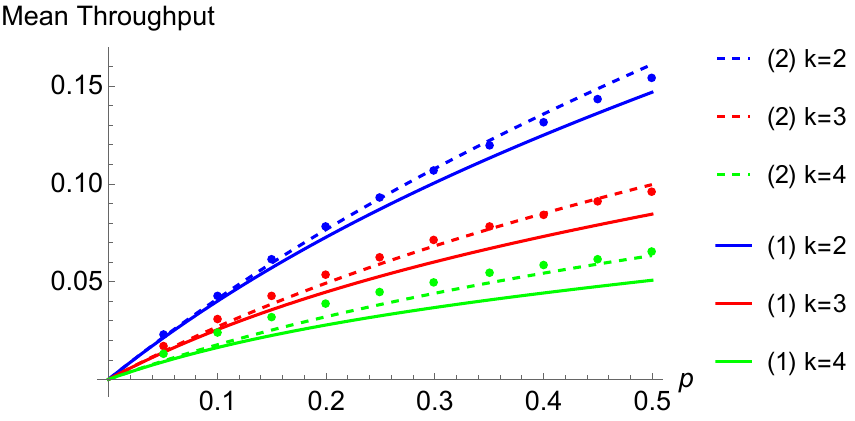}
    \caption{Simulated results and analytical estimations of mean throughput (in $1/\tau$) of $k$-level nested repeater chains.}
    \label{fig:nested_est}
\end{figure}
\section{Conclusion} \label{sec:conclusion}
We formulated quantum repeater network protocols using Markov chain formalism and performed extensive analytical studies of the long-run throughput and the on-demand latency, in single- and two-link quantum repeater chains. The effectiveness of the analytical results is shown to go beyond simple scenarios by comparison with simulation of more complicated scenarios. Another important novel aspect of this work is the study of multiheralded entanglement generation protocols. Our results contribute new insight into the performance of quantum repeater networks in practical scenarios. Moreover, the derived formulae can serve as valuable reference resources for the quantum networking community, for example, as ground truths for quantum network simulation validation, and examples of analytical modeling of quantum network dynamics.

\section*{Acknowledgements}
A. Z. thanks Boxuan Zhou for inspiring discussions. We also thank Naphan Benchasattabuse, Michal Hajdu\v{s}ek, Alexander Kolar, Kento Samuel Soon, Kentaro Teramoto, and Rodney Van Meter for valuable feedback on this work. We acknowledge support from the NSF Quantum Leap Challenge Institute for Hybrid Quantum Architectures and Networks (NSF Award 2016136). This material is based upon work supported by the U.S. Department of Energy, Office Science, Advanced Scientific Computing Research (ASCR) program under contract number DE-AC02-06CH11357 as part of the InterQnet quantum networking project.


\bibliography{references}
\bibliographystyle{IEEEtran}


\end{document}